\documentclass[conference]{IEEEtran}%
\usepackage{amsfonts}
\usepackage{amssymb}
\usepackage{cite}
\usepackage{graphicx}
\usepackage[cmex10]{amsmath}
\usepackage[]{algorithm}
\usepackage{algorithmic}
\usepackage{booktabs}
\usepackage{amsthm}
\usepackage{subfigure} 
\newtheorem{Pro}{Proposition}
\makeatletter

\renewcommand\normalsize{%
	\@setfontsize\normalsize\@xpt\@xiipt
	\abovedisplayskip 6.5\p@ \@plus2\p@ \@minus5\p@
	\abovedisplayshortskip \z@ \@plus3\p@
	\belowdisplayshortskip 6\p@ \@plus3\p@ \@minus3\p@
	\belowdisplayskip \abovedisplayskip
	\let\@listi\@listI}
\makeatother

\ifCLASSINFOpdf
\else
\fi

\begin{document}
\title{Video Streaming in Cooperative Vehicular Networks}

\author{
    \authorblockN{
        Bin Pan\authorrefmark{1}
%        ,
%        Yue Li\authorrefmark{2},
%        Lin Cai\authorrefmark{2},
%        Hao Wu\authorrefmark{1}
    }
    \authorblockA{
        \authorrefmark{1}
        State Key Laboratory of Rail Traffic Control and Safety, Beijing Jiaotong University, Beijing, China}
%    \authorblockA{
%        \authorrefmark{2}
%        Dept. of ECE, University of Victoria, BC, Canada}
}

\maketitle

\begin{abstract}
	Video services in vehicular networks play a significant role in our daily traveling. In this paper, we propose a cooperative communication scheme to facilitate video data transmission, utilizing the mobility of vehicles and the cooperation among infrastructure and vehicles. To improve the video quality of experience (QoE), i.e., reduce the interruption ratio, quality variation and improve the playback quality, we design a Back Compensation (BC) video transmission strategy with the knowledge of vehicle status information. In addition, we analyze the throughput with one-hop and target-cluster-based cooperation schemes and obtain their closed-form expressions, respectively, which is useful for video encoding design in the central server. Simulation results demonstrate that the proposed approach can improve the video performance significantly and verify the accuracy of our analytical results.
\end{abstract}

%\begin{IEEEkeywords}
%
%\end{IEEEkeywords}

% \fi
\IEEEpeerreviewmaketitle

\section{Introduction}

In the past decade, multimedia wireless service has gained significant attention from industry and academic field. Recently, this trend has been extended to vehicular networks thanks to the increasing demand from users in vehicles. Video services may become one of the most popular applications in vehicular networks, from delivering emergency information to providing passengers various entertainment choices. Typically, video services request for sufficient and stable wireless resources due to a large amount of video data to ensure long smooth playback and high video quality. Due to the high cost of deployment, roadside units (RSUs) cannot cover the whole road. When a vehicle moves out of the coverage of RSUs and no other vehicles can act as relays to help connect to RSUs, it cannot receive video data from RSUs directly. Nevertheless, the vehicles moving in the opposite direction (called carriers) can store, carry and forward the data (received from other RSUs) to the target when they meet. In this way, the target vehicle can receive video data continuously along the road, which ensures the smoothness of video playback.

There are many related works focusing on enhancing the video playback quality in vehicular networks. In \cite{xing2012adaptive}, the authors proposed an adaptive video streaming solution, with which a vehicle can download data using a direct or a multi-hop path from the RSUs. The solution includes a simple relay selection algorithm and a video quality adaptation algorithm. In \cite{sun2014quality}, the authors presented a video streaming scheme for cognitive vehicular networks, which can determine the appropriate number of video layers by considering several significant factors including vehicle position, speed and the activity of primary users. In \cite{he2016resource}, with a heterogeneous architecture consisting of the cellular base station and roadside infrastructure, the authors proposed a resource-allocation scheme to facilitate video streaming application based on the semi-Markov decision process (SMDP). A channel allocation and adaptive streaming algorithm was proposed in \cite{sun2017channel} for video services of multiple vehicles, where users compete for wireless channels. The vehicles can bid reasonably according to their utility values through an auction-based mechanism and make the tradeoff between smooth playback and visual quality. Note that the aforesaid schemes utilize the vehicle-to-infrastructure (V2I) communications or relay-aided joint vehicle-to-vehicle (V2V) and V2I communications, i.e., both V2V and V2I links are available at the same time. If the target is far away from RSUs, it will not receive data. Therefore, the target needs the help of carriers through the store-carry-forward strategy.

In this work, we consider scalable video coding (SVC)-based video streaming in cooperative vehicular networks. We utilize a cooperative communication scheme to facilitate video data transmission, which takes advantage of the mobility of vehicles and the cooperation among RSUs and vehicles. A target vehicle will receive video data from carriers in the opposite direction to maintain video playback quality. Then we propose a Back Compensation (BC) video transmission strategy to decrease the interruption ratio, quality variation and improve the playback quality. In addition, we consider target-cluster-based cooperation to increase the amount of data received by the target from carriers. At last, we analyze the throughput, which can give the central server advice for video encoding design. Simulation results show that the performance of our approach is better than that with the relay-aided scheme and greedy video transmission strategy. 

The rest of this paper is organized as follows. Section \uppercase\expandafter{\romannumeral2} presents the system model. Section \uppercase\expandafter{\romannumeral3} describes the BC video transmission strategy. Section \uppercase\expandafter{\romannumeral4} analyzes the throughput of one-hop V2V cooperation and target-cluster-based cooperation respectively. Section \uppercase\expandafter{\romannumeral5} presents the simulation results and conclusions are drawn in section \uppercase\expandafter{\romannumeral6}.

\section{System Model}

\subsection{Network Model}

We focus on the highway scenario with bidirectional vehicle flows, where the roadside infrastructure (e.g., RSUs) are deployed along the road with equal inter-RSU distance $d$, as shown in Fig. \ref{system-model}. It is assumed that the RSUs are all connected to the central server and backbone networks via wireless or wired links. As a result, the vehicles can obtain information services from RSUs while moving within their coverage area.

\begin{figure}[!t]%
	\centering%
	\includegraphics[width=0.45\textwidth]{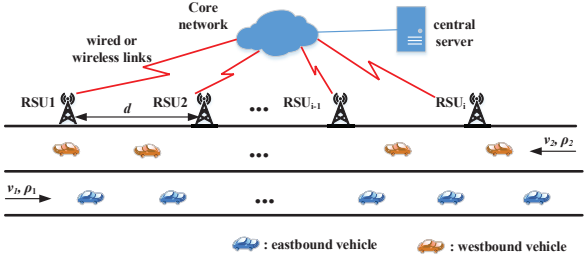}%
	\DeclareGraphicsExtensions. \caption{Network model.}
	\label{system-model}
\end{figure}%
We adopt an extensively used traffic model in a highway scenario \cite{wisitpongphan2007routing}, which in each direction (eastbound and westbound), vehicles arrive following a Poisson process with mean density $\rho_1$ and $\rho_2$, respectively. Therefore, the inter-vehicle distance in each direction follows an exponential distribution. We also assume that the vehicles in each direction move with a constant speed of $v_1$ and $v_2$, respectively \cite{wisitpongphan2007routing}.
\subsection{Wireless Communication Model}
There are two kinds of wireless communication technologies considered in this work: V2V and V2I communications. It is assumed that each vehicle has two antennas: one is for transmitting data and the other is for receiving data. We adopt a simplified but widely used unit disk communication model \cite{wisitpongphan2007routing}. Each vehicle has the same communication range $R_v$ and each RSU has the same communication range $R_u$, with $R_u>R_v$. Because RSUs have a higher capacity than vehicles. Vehicle and RSU (vehicle) can establish the communication link if their distance is less than the communication range $R_u$ ($R_v$).

It is assumed that the data rates of V2I and V2V communication are constant, which are denoted by $r_u$ and $r_v$ respectively. V2I communications can reach a higher data rate compared with V2V communication, i.e., $r_u>r_v$.

%In addition, it is assumed that there is only one channel to allocate to vehicles for video service in V2I and V2V communications respectively, so that there is at most one vehicle to receive data from the same RSU or the same vehicle. Once a vehicle accesses the channel, it will occupy it until it leaves the coverage of the RSU or vehicle. Furthermore, we assume V2I communications are given to a higher priority compared with V2V communications, which means that a vehicle can not obtain data from other vehicles while receiving data from RSU. It is because V2I communications have a higher data rate than V2V communications.

In addition, it is assumed that each target occupies one channel (a frequency band), which is known by the RSU and carriers. All the carriers perform transmission according to the RSU's centralized coordination, i.e., they transmit one by one without collisions. Furthermore, we assume V2I communications are given to a higher priority compared with V2V communications, which means that a vehicle cannot obtain data from other vehicles while receiving data from RSU. It is because V2I communications have a higher data rate than V2V communications.

From the system's point of view, there may be multiple channels, and different targets occupy different channels without contentions. That is to say, one transmitter (RSU or carrier) may transmit to multiple targets at the same time, but on different channels. We assume that the wireless resource is sufficient, so the system's operations, including transmitting and receiving, are independent for each target. Therefore, in this paper, we investigate the procedure of one target. Our analyses can be extended to the multi-target scenario by duplicating the single-target procedure.

\subsection{Cooperative Communication Model}

We assume that each vehicle is equipped with a GPS device, which can obtain its own location information. The vehicles can transmit their basic status information, such as location, speed and direction, to the core network when they enter the coverage of RSUs.

We utilize a cooperative communication scheme to facilitate data transmission while a video is played back. A target can be either a stand-alone vehicle or inside a cluster, which are referred as one-hop and cluster-based scenarios in this paper. In the one-hop scenario, we have the cooperation among RSUs and that between RSU and vehicles. An additional cooperation among vehicles is introduced in the cluster-based scenario. These three cooperation schemes are explained as follows.

\subsubsection{Cooperation among RSUs}
The large video file requested by the target can be divided into multiple segments and pushed into different RSUs. Therefore, each RSU has different segments of video data, which are arranged according to the video playback sequence.

\subsubsection{Cooperation between RSU and vehicles (transmission)}
%Video data delivered to RSUs can be further divided and transmitted to the vehicles when they move within the coverage of RSUs. 

The RSU may not connect to the target all the time, so it will offload some data to the carriers. The RSU determines the data pieces carried by each carrier according to the predicted topology, mobility and the video playback quality. Carriers shall transmit certain data pieces in a fixed time duration according to the RSU's control. By doing so, the RSU and carriers perform the video transmission cooperatively.

%Each carrier has different pieces of data, according to the sequence that target and carriers meet and the sequence of video playback.

\subsubsection{Cooperation among vehicles (receiving)}
%The video data pieces that target received are for its own video playback, and the data pieces that carriers in opposite direction received are carried and forwarded to the target when they meet, which achieves the V2V cooperation taking advantages of the mobility of vehicles and V2V communications.
The kind of cooperation only exists in the cluster-based scenario. All the vehicles in the cluster shall help the target in the same cluster to receive the data. 
%Vehicles in the cluster can transmit and receive video data in different channels simultaneously using V2V technology \cite{ni2016delay}, so the video data can be transmitted as pipeline in cluster.
%We assume perfect transmissions inside a cluster, which means the data received by any vehicle in the cluster is equivalent to that data is successfully received by the corresponding target. In this way, all the vehicles in a cluster receive the video cooperatively.
We assume perfect transmissions inside a cluster, so the data received by any vehicle in the cluster can be transmitted to the corresponding target at data rate $r_v$. In this way, all the vehicles in a cluster receive the video cooperatively.
The formation of a cluster is introduced in Sec.~\ref{cluster_size}.

\subsection{Video Model}
\begin{figure}[!t]%
	\centering%
	\includegraphics[width=0.4\textwidth]{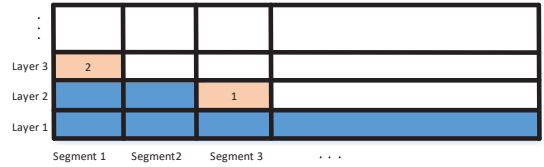}%
	\DeclareGraphicsExtensions. \caption{An example about filling in blocks.}
	\label{example}
\end{figure}%
The whole video can be divided into lots of segments and each segment can be encoded into $L$ video layers based on the SVC technology. Each segment has a base layer and $(L-1)$ enhancement layers. A video segment can be decoded only if the base layer has been received. For video services, the metrics measuring quality of experience (QoE) are interruption ratio, video quality and quality variation \cite{xing2012adaptive}. Interruption ratio is the most important factor for users, which has the highest priority to consider generally. With the guarantee of playback, users wish a better video playback quality. The more layers the target receives, the better video quality it has. In addition, quality variation also has an influence on user experience. For example, as shown in Fig. \ref{example}, the data should be filled in block 1 rather than block 2, although both have the same average video quality. Hence, the blocks in a low layer should be filled in as many as possible before filling in a higher layer. It is assumed that the buffer of each vehicle is large enough to store the video data.

\section{video streaming}

\subsection{One-hop V2V Cooperation}

\begin{figure}[!t]%
	\centering%
	\includegraphics[width=0.45\textwidth]{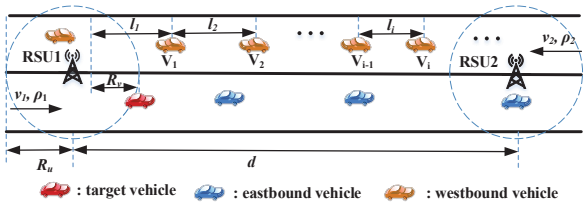}%
	\DeclareGraphicsExtensions. \caption{One-hop V2V cooperation.}
	\label{scenario1}
\end{figure}%

In this section, we introduce our proposed Back Compensation (BC) video transmission strategy. We focus on two adjacent RSUs along the target's moving direction particularly, as shown in Fig. \ref{scenario1}. Denote the first RSU encountered by the target as RSU1 and the second RSU as RSU2. When the target enters the coverage of RSU1, it can receive the video data directly from RSU1. When it moves out of the coverage of RSU1, it can keep receiving data from carriers through V2V communications since each carrier has received different segments of video data from RSU2. In this scenario, we first consider one-hop V2V communications between the target and carriers. Similar procedures repeat as the target travels along the road.

Through the cooperative communication scheme, the core network have the knowledge of carriers that the target will meet. Therefore, we can calculate the amount of data received by the target from RSU1 and carriers.

As for V2I communications, the amount of data the target receives from RSU1 can be expressed as $\frac{{2{R_u}{r_u}}}{{{v_1}}}$. As for V2V communications, a carrier must receive video data from RSU2 before forwarding the data to the target. Therefore, there may be a case that a carrier does not carry sufficient data, which means the carrier has transmitted all of carried data to the target before leaving the coverage of the target. We can prove that such a case will not happen in this scenario, which is shown in the proposition below.

\begin{Pro}
	In the duration of V2V communications, all carriers the target will meet have sufficient amount of data to transmit, with the assumption of $R_u>R_v$ and $r_u>r_v$.
\end{Pro}

\begin{proof}
	The amount of data each carrier $V_i$ receives from RSU2 is expressed as 
	\begin{equation}\label{}
	D_i^u = \min \{ \frac{{{l_i}}}{{{v_2}}}{r_u},\frac{{2{R_u}}}{{{v_2}}}{r_u}\} ,i = 1,2,...,
	\end{equation}
	where $l_i$ is the inter-vehicle distance between $V_{i-1}$ and $V_i$. The amount of data the target receives from each carrier $V_i$ without considering the limitation of the amount of data they receive from RSU2 can be expressed as
	\begin{equation}\label{}
	D_i^v = \left\{ \begin{array}{l}
	\min \{ \frac{{{l_{1'}}}}{{{v_1} + {v_2}}}{r_v},\frac{{2{R_v}}}{{{v_1} + {v_2}}}{r_v}\} ,i = 1,\\
	\min \{ \frac{{{l_i}}}{{{v_1} + {v_2}}}{r_v},\frac{{2{R_v}}}{{{v_1} + {v_2}}}{r_v}\} ,i = 2,3,...,
	\end{array} \right.
	\end{equation}
	where ${l_1} \ge {l_{1'}}$. Thus we have $D_i^v<D_i^u, i=1,2,...$, considering $R_u>R_v$, $r_u>r_v$. Therefore, we can conclude that all carriers have sufficient amount of data to transmit to the target.

\end{proof}

Therefore, we can obtain the amount of data the target receives from all carriers with V2V communications as
\begin{equation}\label{}
{D_i} = \min \{ \frac{{{l_i}}}{{{v_1} + {v_2}}}{r_v},\frac{{2{R_v}}}{{{v_1} + {v_2}}}{r_v}\} ,i = 1,2,...,
\end{equation}
and $l_1$ here denotes the distance as shown in Fig. \ref{scenario1} at the moment that the target moves out of the coverage of RSU1. $l_i$ follows exponential distribution according to the memoryless property.

We can also calculate the time instant when the target meets carriers according to their locations and speeds. Having known the amount of data the target receives at any moment, now we determine which data should be transmitted. We define the bit rate of each video layer as $r_q^b$, $q=1,2,...,L$.
Denote the carriers the target will meet in sequence as $V_i$, $i=1,2,...$. Denote the amount of data received by the target as $D_i, i=0,1,2...$, where $D_0$ and $D_i$ ($i=1,2,...$) means the amount of data from RSU1 and carriers $V_i$ respectively. Denote the remaining amount of data after data filling in the $q$th layer as $D_i^R$. Denote the time interval between the time instant when the target begins to receive $D_i$ and the time instant when the target begins to receive $D_{i+1}$ as $t_i, i=0,1,2,...$. Denote the maximum video playback duration of the $i$th time interval in the $q$th layer as $t_{q,i}$ and denote the actual playback duration of $t_i$ in the $q$th layer as $\Delta t_{q,i}'$.

\begin{figure}[!t]%
	\centering%
	\includegraphics[width=0.48\textwidth]{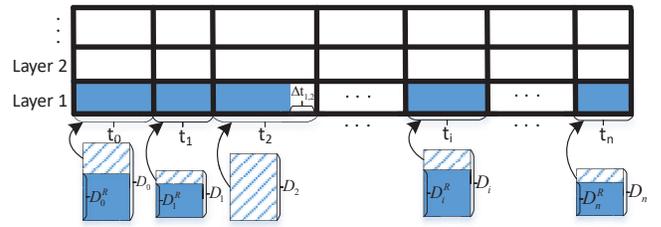}%
	\DeclareGraphicsExtensions. \caption{Video transmission design.}
	\label{video_design}
\end{figure}%

As we mentioned above, we should put data into the buffer of the base layer first. As shown in Fig. \ref{video_design}, there are totally ($n+1$) buffer blocks in the base layer, where $n$ means the number of carriers the target will meet and the buffer blocks correspond to the time interval $t_i$s. For the base layer, $D_i^R$ is equal to $D_i$. Put data $D_i^R$ into buffer block $t_{1,i}$ respectively. As a result, some buffers like $t_{1,0}$ and $t_{1,1}$ are full, but some are not full like $t_{1,2}$. Hence, there may be some interruption time intervals, like $\Delta t_{1,2}$ in $t_{1,2}$. Considering the property of video playback, the arrival data in the enhancement layer is invalid if its segment has been played. Therefore, we fill in the interruption segments from back to front (from $t_n$ to $t_0$). In addition, for each interruption segment $t_i$, the remaining data we choose to fill in are also in the sequence from back to front starting from $D_{i-1}^R$. Applying this strategy from the low layer to the higher layers until all $D_i^R$s are equal to 0 or the buffer block $t_{L,n}$ is full. At last, we adopt the \textit{dress right} operation for each layer to decrease the layer variation, which means to push the data filled in the buffer block backward. As a result, the largest quality variation is equal to $L$. 

Algorithm 1 shows the details of the video transmission design. Line 2 to 8 describe the process of filling in the $t_{q,i}$ buffer block with remaining data $D_i^R$ respectively. Line 9 to 20 describe the process of filling in the interruption block from back to front. The maximum video playback duration of the $i$th time interval in the ($q+1$)th layer can be obtained from line 21 to 23. Line 24 to 27 describe the conditions to break the loop. Finally, execute the \textit{dress right} operation for each layer.

After the process of video transmission design, we allocate the data to RSU1 and carriers in turn, as shown in Fig. \ref{video_allocation}.

\begin{algorithm}[!t]
	\caption{Video Transmission Design}
	\begin{algorithmic}[1]
		%		\STATE \textbf{Input}: Accumulated video data $D_i$ and speed $v_i$ of new vehicle $i$, Urgent set $S_u$, Current time slot index $t$.
		%		\STATE \textbf{Output}: Decision (admit or decline the new vehicle).
		%		\STATE \textbf{procedure} ACA ($D_i$, $v_i$, $S_u$, $t$)
		%$BR_q$: the video data rate of the $q$th layer.
		%$t_i^q$: the time interval 
		%$D_i^R$: the remaining amount of data received from $v_i$.
		\FOR {$q=1$ to $L$}
		\FOR {$i=0$ to $n$}
		\IF {$D_i^R\geq r_q^b\cdot t_{q,i}$}
		\STATE $t_{q,i}' \leftarrow t_{q,i}$, $D_i^R \leftarrow D_i^R-r_q^b \cdot t_{q,i}'$.			 
		\ELSE
		\STATE $t_{q,i}' \leftarrow \frac{D_i^R}{r_q^b}$, $D_i^R \leftarrow 0$.				 
		\ENDIF		
		\ENDFOR
		\FOR{$i=n$ to $0$}
		\IF{$t_{q,i}'<t_{q,i}$}
		\FOR{$j=i-1$ to $0$}
		
		\IF{$D_j^R\geq (t_{q,i}-t_{q,i}')\cdot r_q^b$}
		\STATE $D_j^R\leftarrow D_j^R-(t_{q,i}-t_{q,i}')\cdot r_q^b$, $t_{q,i}' \leftarrow t_{q,i}$,
		\STATE Break the loop.
		\ELSE 
		\STATE $t_{q,i}' \leftarrow t_{q,i}'+ \frac{D_j^R}{r_q^b}$, $D_j^R\leftarrow 0$.
		
		\ENDIF
		
		\ENDFOR
		\ENDIF
		\ENDFOR
		\FOR{$i=0$ to $n$}
		\STATE $t_{q+1,i} \leftarrow t_{q,i}'$.
		\ENDFOR
		\IF{All $D_i^R$ are equal to 0 or the buffer block $t_{L,n}$ is full}
		\STATE Break the loop.
		\ENDIF
		\ENDFOR
		\STATE Dress right for each layer. 
	\end{algorithmic}
\end{algorithm}
\begin{figure}[!t]%
	\centering%
	\includegraphics[width=0.4\textwidth]{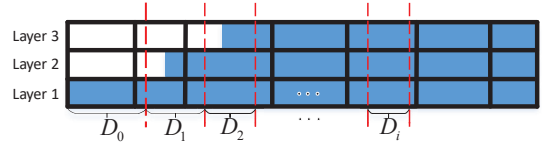}%
	\DeclareGraphicsExtensions. \caption{Video data allocation.}
	\label{video_allocation}
\end{figure}%
\subsection{Target Cluster-based Cooperation}
\label{cluster_size}

In the previous subsection, the amount of data received by the target from carriers are limited by V2V communications. However, the target can receive more video data  by increasing the hops of V2V communications. In vehicular networks, vehicles moving in the same direction can form clusters, in which the distances between two neighboring vehicles are no more than V2V communication range $R_v$. Video data can be transmitted to other vehicles through multi-hop links in a cluster. Therefore, we also consider target-cluster-based cooperation as shown in Fig. \ref{scenario2}.

Denote the cluster size as $C_s$, which means the distance between the first and last vehicles in the cluster. So the amount of data the target receives from RSU1 directly and indirectly (with the help of cluster relays) can be expressed as $\frac{{{C_s} + 2{R_u}}}{{{v_1}}}$. When the target has received all data from RSU1, it relies on carriers in the opposite direction, which carry and forward data from RSU2. We also consider the simple First-In-First-Out (FIFO) scheduling scheme for V2I and V2V communication mentioned above. This scheme is not the optimal solution, because the optimal schedule should consider both the correlation between the amount of data received by neighboring carriers from RSU2 and between the amount of data received by the target from neighboring carriers, which is much complicated and left for future work. Therefore, the amount of data the target can receive from each carrier is expressed as

\begin{small}
	\begin{equation}\label{}
	{D_i^c} = \left\{ \begin{array}{l}
	\min \{ \frac{{{l_1}}}{{{v_2}}}{r_u},\frac{{2{R_u}}}{{{v_2}}}{r_u},\frac{{{l_{1'}}}}{{{v_1} + {v_2}}}{r_v},\frac{{{C_s} + 2{R_v}}}{{{v_1} + {v_2}}}{r_v}\} ,i = 1,\\
	\min \{ \frac{{{l_i}}}{{{v_2}}}{r_u},\frac{{2{R_u}}}{{{v_2}}}{r_u},\frac{{{l_i}}}{{{v_1} + {v_2}}}{r_v},\frac{{{C_s} + 2{R_v}}}{{{v_1} + {v_2}}}{r_v}\} ,i = 2,3,...,
	\end{array} \right.
	\end{equation}
\end{small}
where ${l_1} \ge {l_{1'}}$. Therefore, according to $R_u>R_v$, $r_u>r_v$ and the memoryless property of exponential distribution, (4) can be simplified to
\begin{equation}\label{}
{D_i^c} = \min \{ \frac{{2{R_u}}}{{{v_2}}}{r_u},\frac{{{l_i}}}{{{v_1} + {v_2}}}{r_v},\frac{{{C_s} + 2{R_v}}}{{{v_1} + {v_2}}}{r_v}\} ,i = 1,2,3,...,
\end{equation}
and $l_1$ here denote the distance as shown in Fig. \ref{scenario2} at the time instant when the target has received all data from RSU1 with the help of relays in cluster. Given a cluster size $C_s$, we can use the same BC strategy mentioned in one-hop V2V cooperation.

\begin{figure}[!t]%
	\centering%
	\includegraphics[width=0.45\textwidth]{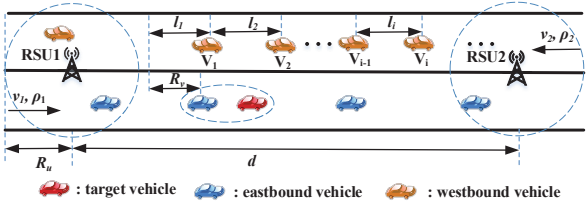}%
	\DeclareGraphicsExtensions. \caption{Target-cluster-based cooperation.}
	\label{scenario2}
\end{figure}%

\section{Throughput Analysis}

In this section, we make the analysis of throughput, which is defined as the amount of data received by the target per unit time. If the bit rate of base layer is higher than the throughput, it is likely to cause a long interruption time. If the total bit rate of $L$ layer is lower than the throughput, it may cause the waste of wireless resource. Therefore, the value of throughput can give the central server advice for video encoding design.
\subsection{One-hop V2V Cooperation}

We denote one period as the interval between the time instant when the target enters in the RSU1's coverage and the time instant when the target enters in the RSU2's coverage, which is $T=d/v_1$. Hence, the throughput in one period can be given by:
\begin{equation}\label{}
\eta {\rm{ = }}\frac{{E[{D_{u}}] + E[{D_{v}}]}}{T},
\end{equation}
where $E[D_{u}]$ and $E[D_{v}]$ represent the expected amount of data the target receives from RSU1 and carriers in one period respectively.

When the target is within the coverage of RSU1, it can only receive data through V2I communications. Therefore, the expected amount of data the target receives from RSU1 is
\begin{equation}\label{}
E[{D_{u}}] = \frac{{2{R_u}{r_u}}}{{{v_1}}}.
\end{equation}

Denote the number of carriers the target will meet in one period as $N$, which is a random non-negative integer. As mentioned above, the amount of data the target can receive from a carrier $V_i$ is expressed as
\begin{equation}\label{}
{D_{i}} = \frac{{\min \{ {l_i},2{R_v}\} }}{{{v_1} + {v_2}}} \cdot {r_v},i = 1,2,...,{N}.
\end{equation}
Hence, the total amount of data the target receives from all carriers in one period, denoted by $D_{v}$, can be expressed as the summation of the amount of data received by the target from $N$ carriers:
\begin{equation}\label{}
{D_{v}} = \sum\limits_{i = 1}^{{N}} {{D_{i}}}.
\end{equation}

Note that both the number of carriers $N$ and the inter-vehicle distance $l_i$ are random non-independent variables. Like \cite{chen2017throughput}, we approximately assume that these two variables are independent. Therefore, the expected amount of data the target receives from carriers in one period, denoted by $E[D_{v}]$ can be expressed as:
\begin{equation}\label{}
{E[D_{v}]} \approx E[{N}] \cdot E[{D_{i}}],
\end{equation}
where $E[N]$ is the expected number of carriers the target can meet in one period, and $E[D_i]$ is the expected amount of data the target received from a carrier. According to \cite{gallager2013stochastic}, when $(d-2R_u)/v_1\gg E[t_{inter}]$, $E[N]$ can be approximated by:
\begin{equation}\label{}
E[{N}] \approx \frac{{d - 2{R_u}}}{{{v_1} \cdot E[{t_{inter}}]}},
\end{equation}
where $E[t_{inter}]$ is the expected arrival time interval of carriers, which can be calculated as
\begin{equation}\label{}
E[{t_{inter}}] = \frac{1}{{{\rho _2}({v_1} + {v_2})}},
\end{equation}
where $(v_1+v_2)$ means the relative velocity between the target and carriers.
The expected amount of data between the target and a carrier $V_i$, denoted by $E[D_i]$, can be calculated as
%\begin{small}
%\begin{equation}\label{}
%\begin{split}
%E[{D_{i1}}]&= \frac{{\Pr \{ {l_i} < 2{R_v}\}  \cdot E[{l_i}|{l_i} < 2{R_v}] + \Pr \{ {l_i} \ge 2{R_v}\}  \cdot 2{R_v}}}{{{v_1} + {v_2}}} \cdot {r_v}\\
%&= \frac{{(1 - {e^{ - 2{\rho _2}{R_v}}}) \cdot (\frac{1}{{{\rho _2}}} - \frac{{2{R_v}{e^{ - 2{\rho _2}{R_v}}}}}{{1 - {e^{ - 2{\rho _2}{R_v}}}}}) + {e^{ - 2{\rho _2}{R_v}}} \cdot 2{R_v}}}{{{v_1} + {v_2}}} \cdot {r_v}
%\end{split}
%\end{equation}
%\end{small}

\begin{equation}\label{}
\begin{split}
E[{D_{i}}] &= (\Pr \{ {l_i} < 2{R_v}\}  \cdot E[{l_i}|{l_i} < 2{R_v}]\\
& \quad {\rm{ + }}\Pr \{ {l_i} \ge 2{R_v}\}  \cdot 2{R_v}) \cdot \frac{{{r_v}}}{{{v_1} + {v_2}}} \\
&= ((1 - e{}^{ - 2{\rho _2}{R_v}}) \cdot (\frac{1}{{{\rho _2}}} - \frac{{2{R_v}{e^{ - 2{\rho _2}{R_v}}}}}{{1 - {e^{ - 2{\rho _2}{R_v}}}}}) \\
& \quad + e{}^{ - 2{\rho _2}{R_v}} \cdot 2R_v) \cdot \frac{{{r_v}}}{{{v_1} + {v_2}}}.
\end{split}
\end{equation}

Hence, we can obtain $\eta$ combining (6), (7), (10-13).

\subsection{Target-cluster-based Cooperation}

According to \cite{he2017delay}, the probability distribution function of cluster size can be given by:
\begin{equation}\label{}
{f_c}(x) = \frac{{{x^{k - 1}}{e^{ - \frac{x}{\theta }}}}}{{{\theta ^k}\Gamma (k)}},x > 0,
\end{equation}
where $k = \frac{{E{{\{ {C_s}\} }^2}}}{{E\{ {C_s}^2\}  - E{{\{ {C_s}\} }^2}}}$ and $\theta  = \frac{{E\{ {C_s}^2\}  - E{{\{ {C_s}\} }^2}}}{{E\{ {C_s}\} }}$, $\Gamma ( \cdot )$ is the Gamma function, and $E\{C_s\}$ and $E\{ {C_s}^2\}$ can be obtained from \cite{he2017delay}.

\begin{figure}[!t]%
	\centering%
	\includegraphics[width=0.45\textwidth]{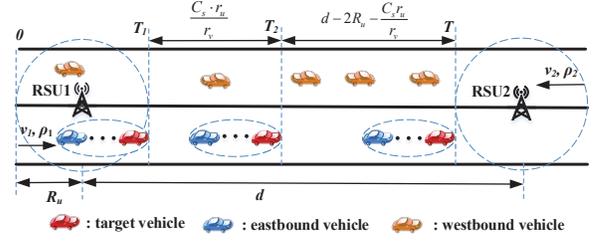}%
	\DeclareGraphicsExtensions. \caption{Definition of $T_1$ and $T_2$.}
	\label{define_T}
\end{figure}%

The amount of data the target received from each carrier $V_i$ can be expressed as (5). 
If it satisfies $\frac{{2{R_u}}}{{{v_2}}}{r_u} > \frac{{2{R_v} + {C_s}}}{{{v_1} + {v_2}}}{r_v}$, i.e., ${C_s} < \frac{{2{R_u}{r_u}}}{{{v_2}{r_v}}}({v_1} + {v_2}) - 2{R_v}$, ${D_{i}^c} = \frac{{\min \{ {l_i},{C_s} + 2{R_v}\} }}{{{v_1} + {v_2}}}{r_v}$, otherwise, ${D_{i}^c} = \min \{ \frac{{2{R_u}}}{{{v_2}}}{r_u},\frac{{{l_i}}}{{{v_1} + {v_2}}}{r_v}\}$. However, if the length of cluster $C_s$ is large enough, the target will not need the help of V2V communications. As shown in Fig. \ref{define_T}, denote the time instant when the target moves out of the coverage of RSU1 as $T_1$, and the time instant when the amount of data the target cluster receives from RSU1 all forwarded to the target as $T_2$. The distance that the target moves from $T_1$ to $T_2$ can be expressed as $(\frac{{{C_s}}}{{{v_1}}}{r_u})/{r_v} \cdot {v_1} = \frac{{{C_s}{r_u}}}{{{r_v}}}$. Therefore, if the target needs the help of carriers through V2V communications, it should satisfy the condition $d - 2{R_u} - \frac{{{C_s}{r_u}}}{{{r_v}}} > 0$, i.e., ${C_s} < \frac{{(d - 2{R_u}){r_v}}}{{{r_u}}}$. Hence, we denote two threshold values of $C_s$ as follows:
\begin{equation}\label{}
{C_{s1}} = \frac{{2{R_u}{r_u}}}{{{v_2}{r_v}}}({v_1} + {v_2}) - 2{R_v},
\end{equation}
\begin{equation}\label{}
{C_{s2}} = \frac{{(d - 2{R_u})}}{{{r_u}}}{r_v}.
\end{equation}

If $C_{s1}\leq C_{s2}$, i.e., $d \geq (\frac{{2{R_u}{r_u}}}{{{v_2}{r_v}}}({v_1} + {v_2}) - 2{R_v})\frac{{{r_u}}}{{{r_v}}} + 2{R_u}$, there are three cases below.

\textit{Case 1.1}: $0\leq C_s\leq C_{s1}$. In this case, given a certain $C_s$, the throughput in one period can be given by:
\begin{equation}\label{}
{\eta _{11}^c}{\rm{ = }}\frac{{E[{D_{u11}^c}] + E[{D_{v11}^c}]}}{T},
\end{equation}
where $E[D_{u11}^c]$ represents the expected amount of data the target receives from RSU1 directly and indirectly, and $E[D_{v11}^c]$ represents the expected amount of data the target receives from the carriers. $E[D_{u11}^c]$ can be expressed as:
\begin{equation}\label{}
E[{D_{u11}^c}] = \frac{{2{R_u} + {C_s}}}{{{v_1}}}{r_u}.
\end{equation}
The expected number of carriers the target can meet in one period, denoted by $E[N_{11}^c]$, can be expressed as
\begin{equation}\label{}
E[{N_{11}^c}] \approx \frac{{d - 2{R_u} - \frac{{{C_s}{r_u}}}{{{r_v}}}}}{{{v_1} \cdot E[{t_{inter}}]}}.
\end{equation}
In this case, ${D_{i}^c} = \frac{{\min \{ {l_i},{C_s} + 2{R_v}\} }}{{{v_1} + {v_2}}} \cdot {r_v}$. Hence, the expected amount of data the target receives from a carrier can be expressed as
	\begin{equation}\label{}
	\begin{split}
	E[D_{i11}^c]  &= (\Pr \{ {l_i} < {C_s} + 2{R_v}\}  \cdot E[{l_i}|{l_i} < {C_s} + 2{R_v}] \\
	&\quad + \Pr \{ {l_i} \ge {C_s} + 2{R_v}\}  \cdot ({C_s} + 2{R_v})) \cdot \frac{{{r_v}}}{{{v_1} + {v_2}}},
	\end{split}
	\end{equation}

where
\begin{equation}\label{}
\Pr \{ {l_i} < {C_s} + 2{R_v}\}  = 1 - {e^{ - {\rho _2}({C_s} + 2{R_v})}},
\end{equation}
\begin{equation}\label{}
\Pr \{ {l_i} \ge {C_s} + 2{R_v}\}  = {e^{ - {\rho _2}({C_s} + 2{R_v})}},
\end{equation}
\begin{equation}\label{}
E[{l_i}|{l_i} < {C_s} + 2{R_v}] = \frac{1}{{{\rho _2}}} - \frac{{({C_s} + 2{R_v}){e^{ - {\rho _2}({C_s} + 2{R_v})}}}}{{1 - {e^{ - {\rho _2}({C_s} + 2{R_v})}}}}.
\end{equation}

Therefore, the expected amount of data $E[D_{v11}^c]$ can be expressed as
\begin{equation}\label{}
E[{D_{v11}^c}] \approx E[N_{11}^c] \cdot E[D_{i11}^c].
\end{equation}

\textit{Case 1.2}: $C_{s1}<C_s\leq C_{s2}$.
In this case, similarly to Case 1.1, given a certain $C_s$, the throughput in one period can be given by:
\begin{equation}\label{}
{\eta _{12}^c}{\rm{ = }}\frac{{E[{D_{u12}^c}] + E[{D_{v12}^c}]}}{T}.
\end{equation}
Obviously, $E[D_{u12}^c]$ is equal to $E[D_{u11}^c]$, and $E[N_{12}^c]$ is equal to $E[N_{11}^c]$. In this case, $D_{i}^c = \min \{ \frac{{{l_i}}}{{{v_1} + {v_2}}}{r_v},\frac{{2{R_u}}}{{{v_2}}}{r_u}\}$, and the expected amount of data the target receives from a carrier can be expressed as
\begin{equation}\label{}
\begin{split}
E[{D_{i12}^c}] &= \Pr \{ {l_i} < c_1\}  \cdot E[{l_i}|{l_i} < c_1] \cdot \frac{{{r_v}}}{{{v_1} + {v_2}}} \\
&\quad + \Pr \{ {l_i} \ge c_1\}  \cdot \frac{{2{R_u}}}{{{v_2}}}{r_u},
\end{split}
\end{equation}
where 
\begin{equation}\label{}
c_1 = \frac{{2{R_u}{r_u}({v_1} + {v_2})}}{{{r_v}{v_2}}},
\end{equation}
\begin{equation}\label{}
\Pr \{ {l_i} < c_1\}  = 1 - {e^{ - {\rho _2}c_1}},
\end{equation}
\begin{equation}\label{}
\Pr \{ {l_i} \ge c_1\}  = {e^{ - {\rho _2}c_1}},
\end{equation}
\begin{equation}\label{}
E[{l_i}|{l_i} < c_1] = \frac{1}{{{\rho _2}}} - \frac{{c_1{e^{ - {\rho _2}c_1}}}}{{1 - {e^{ - {\rho _2}c_1}}}}.
\end{equation}
Therefore, the expected amount of data $E[D_{v12}^c]$ can be expressed as
\begin{equation}\label{}
E[{D_{v12}^c}] \approx E[N_{12}^c] \cdot E[D_{i12}^c]. 
\end{equation}

\textit{Case 1.3}: $C_s>C_{s2}$. In this case, given a certain $C_s$, the throughput in one period can be given by:
\begin{equation}\label{}
{\eta _{13}^c} = \frac{{E[{D_{u13}^c}]}}{T},
\end{equation}
where
\begin{equation}\label{}
E[{D_{u13}^c}] = \frac{{2{R_u}}}{{{v_1}}}{r_u} + \frac{{(d - 2{R_u})}}{{{v_1}}}{r_v}.
\end{equation}
Therefore, the throughput considering Case 1.1,  Case 1.2 and Case 1.3  can be calculated by:
\begin{equation}\label{}
\begin{split}
{\eta _1^c} &= \int_0^{{C_{s1}}} {{\eta _{11}^c} \cdot {f_c}({C_s})d{C_s}} + \int_{{C_{s1}}}^{{C_{s2}}} {{\eta _{12}^c} \cdot {f_c}({C_s})d{C_s}}  \\
&\quad + \int_{{C_{s2}}}^{ + \infty } {{\eta _{13}^c}}  \cdot {f_c}({C_s})d{C_s}.
\end{split}
\end{equation}

If $C_{s1}> C_{s2}$, i.e., $d < (\frac{{2{R_u}{r_u}}}{{{v_2}{r_v}}}({v_1} + {v_2}) - 2{R_v})\frac{{{r_u}}}{{{r_v}}} + 2{R_u}$, there are two cases below.

\textit{Case 2.1}: $0\leq C_s\leq C_{s2}$. This case is similar to Case 1.1, so given a certain $C_s$, the throughput in one period $\eta_{21}^c$ is equal to $\eta_{11}^c$. 

\textit{Case 2.2}: $C_s>C_{s2}$. This case is similar to Case 1.3, so given a certain $C_s$, the throughput in one period $\eta_{22}^c$ is equal to $\eta_{13}^c$.

Therefore, the throughput considering Case 2.1 and Case 2.2 can be calculated by:
\begin{equation}\label{}
{\eta _2^c} = \int_0^{{C_{s2}}} {{\eta _{21}^c} \cdot {f_c}({C_s})d{C_s}} + \int_{{C_{s2}}}^{ + \infty } {{\eta _{22}^c}}  \cdot {f_c}({C_s})d{C_s}.
\end{equation}

\section{Performance Evaluation}
\begin{figure*}[!t]
	\centering
	\begin{minipage}[t]{0.25\linewidth}
		\centering
		\includegraphics[width=1\textwidth]{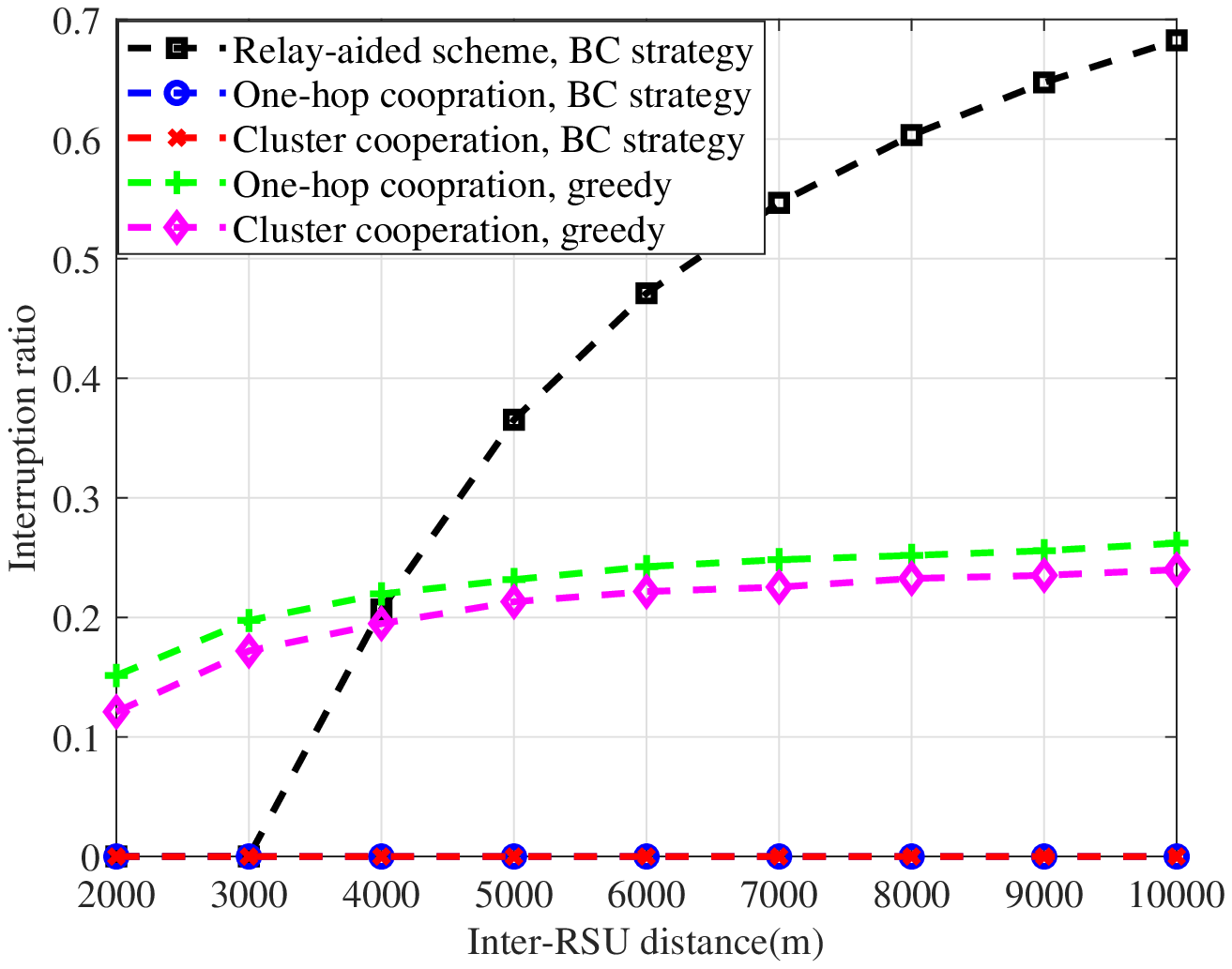}
		\caption{Interruption ratio.}
		\label{IR}
	\end{minipage}%
	\begin{minipage}[t]{0.25\linewidth}
		\centering
		\includegraphics[width=1\textwidth]{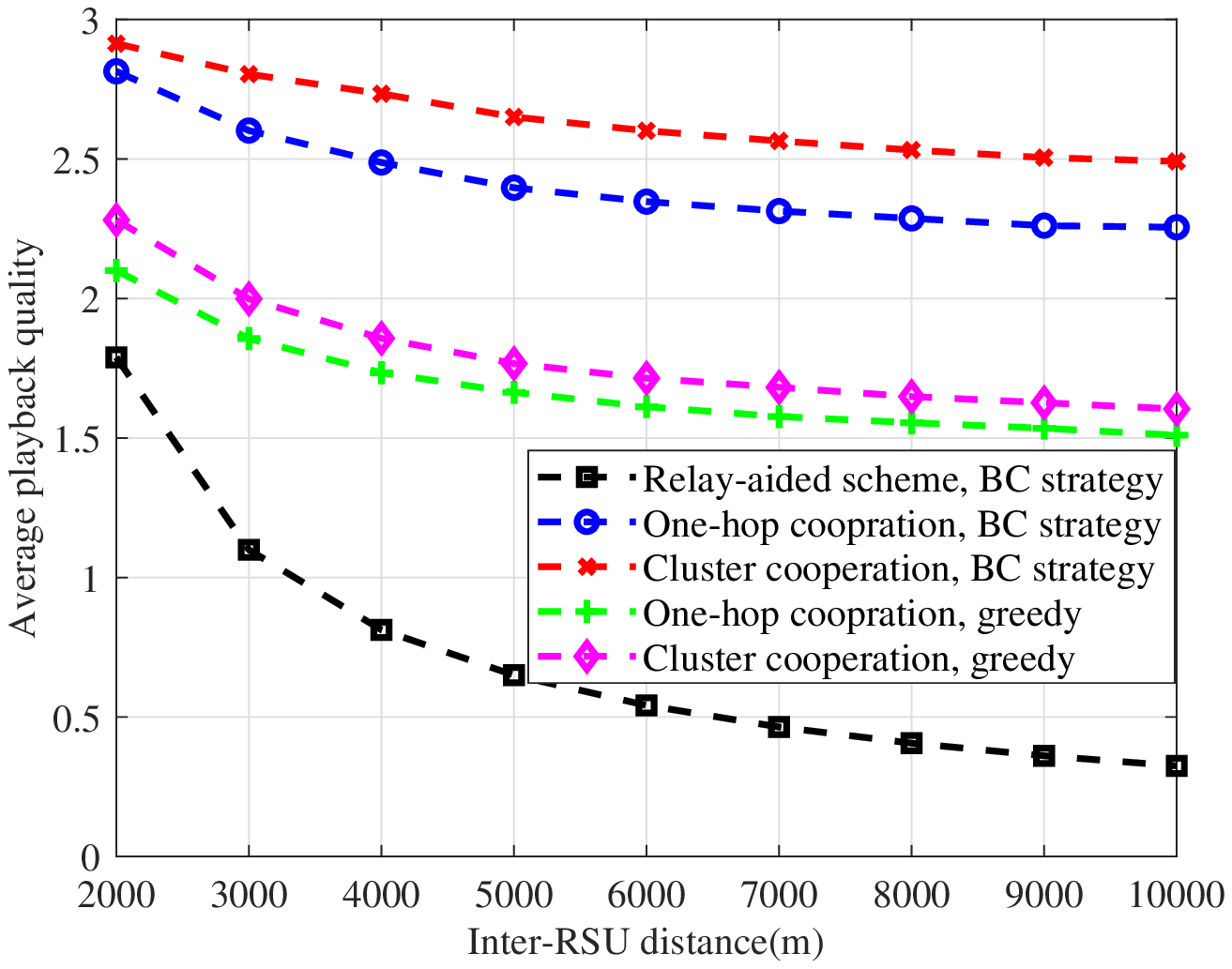}
		\caption{Average playback quality.}
		\label{Ave_qua}
	\end{minipage}%
	\begin{minipage}[t]{0.25\linewidth}
		\centering
		\includegraphics[width=1\textwidth]{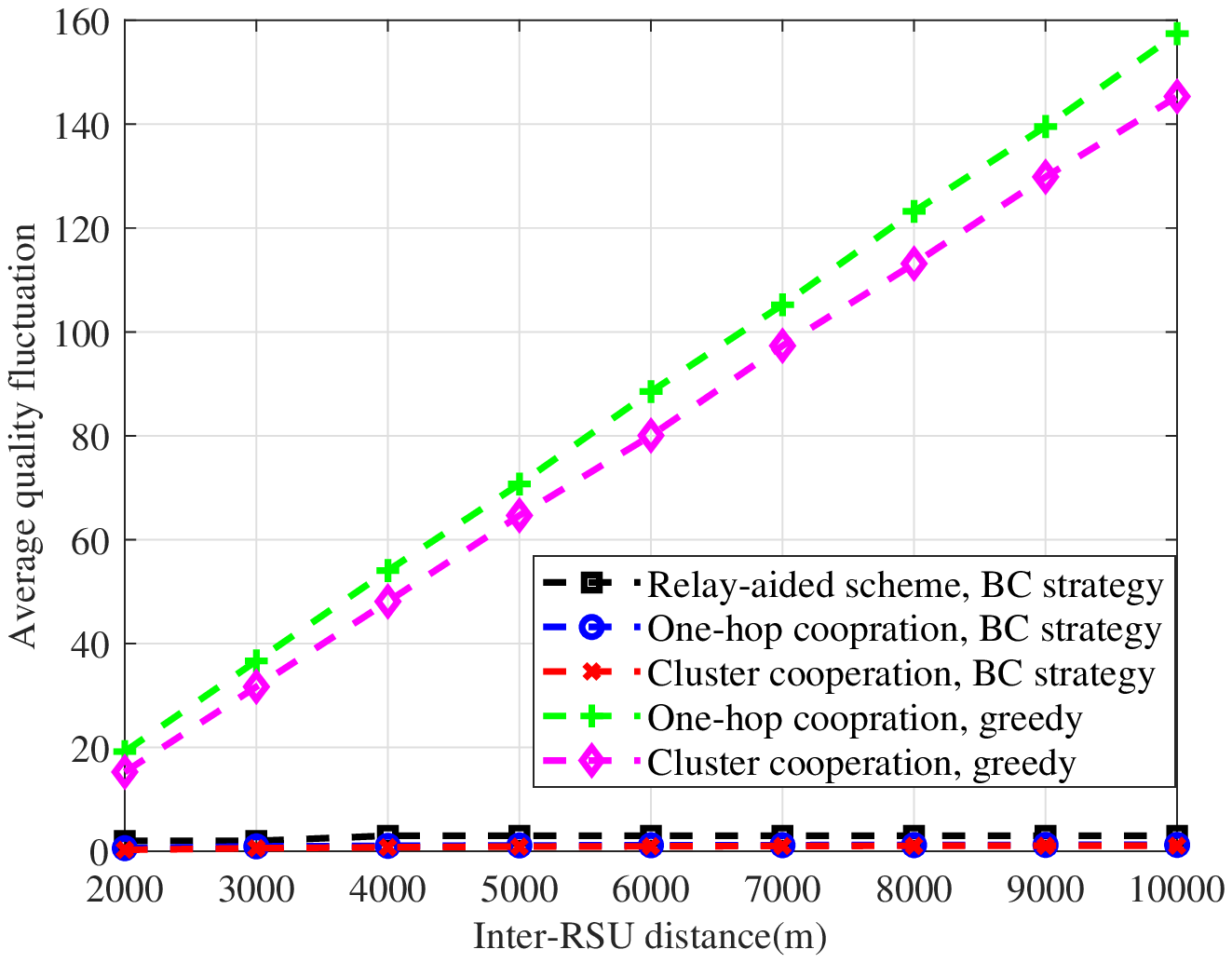}
		\caption{Average quality variation.}
		\label{Ave_flu}
	\end{minipage}%
	\begin{minipage}[t]{0.25\linewidth}
		\centering
		\includegraphics[width=1\textwidth]{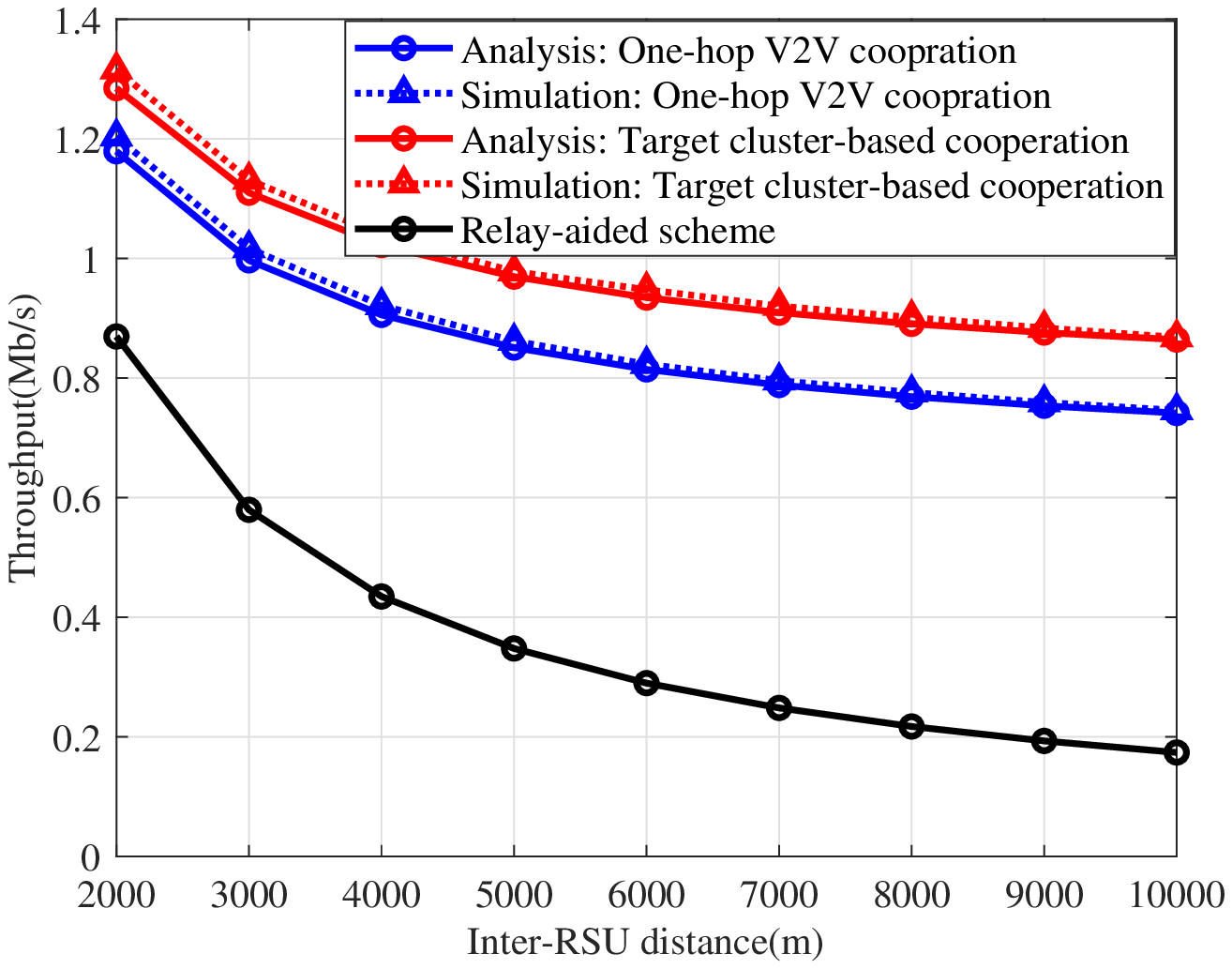}
		\caption{Throughput.}
		\label{throughput}
	\end{minipage}
\end{figure*}

%\begin{table}
%	\caption{Video Encoding Configurations.}
%	\label{Video_para}
%	\centering
%%	\newcommand{\tabincell}[2]{\begin{tabular}{@{}#1@{}}#2\end{tabular}}
%	\begin{tabular}{|c|c|c|}
%	\hline
%	Layer Index&Ave. bit-rate(Kbps)&PSNR(dB)\\
%	\hline
%	1&537.9&29.16\\
%	\hline
%	2&958.2&33.27\\
%	\hline
%	3&1408.7&35.53\\
%	\hline
%	\end{tabular}
%\end{table}
To evaluate the performance of our proposed cooperation scheme and BC video transmission strategy, we implemented simulation with MATLAB. We consider the bidirectional highway with eastbound and westbound lanes, where the vehicle densities are $\rho_1$ = 0.007veh/m and $\rho_2$ = 0.005veh/m and the velocities are $v_1$ = 15m/s and $v_2$ = 20m/s respectively. The RSUs are regularly deployed along the road with inter-RSU distance $d$ varying from 2000m to 10000m. The communication range of RSUs and vehicles are $R_u$ = 400m and $R_v$ = 200m (typical communication range with DSRC \cite{chen2016achievable}), respectively. The bit rate of V2V and V2I communications are $r_v$ = 1Mb/s and $r_u$ = 2Mb/s. For each simulation result, it is an average of 2000 trials with random seeds. We use a real video trace of $Football$ \cite{xing2016maximum} in our simulation, which is encoded into a base layer and two enhancement layers. The average bit rate of the three video layers are 537.9Kbps, 420.3Kbps and 450.5Kbps, respectively.

We use three metrics to quantify the performance. \textbf{Interruption ratio}: Interruption ratio (IR) is defined as the ratio of video interruption time to the target running time at this road segment, which can denoted by $IR=1-\sum\nolimits_{i = 0}^n {{t_{1,i}'}} /T $. \textbf{Average playback quality}: As mentioned in \cite{xing2012adaptive}, the video with lower quality may have a higher PSNR value for SVC technology. Hence, we define the average playback quality (APQ) as the average layer number to all playback segments, which is denoted by $APQ=\sum\nolimits_{q = 1}^L {\sum\nolimits_{i = 0}^n {{t_{q,i}'}} } /T$. \textbf{Average quality variation}: Average quality variation (AQV) is defined as the number of video layer changing, from low layer to high layer or from high layer to low layer.

The performance results are shown in Figs. \ref{IR}, \ref{Ave_qua} and \ref{Ave_flu}. We compare the results with relay-aided scheme and greedy video transmission strategy. \textit{Gr} means using the amount of data $D_i$ to service time interval $t_i$. From Fig. \ref{IR}, we can see that the interruption ratio with cooperation schemes and BC strategy can be always maintained an extremely small value compared with others, because the target can obtain the help of carriers continuously along the road and the data received by the target earlier can be used to fill in subsequent blocks. From Fig. \ref{Ave_qua}, we can see that the cooperation schemes can achieve a higher average playback quality than the relay-aided scheme. In addition, the performance of our BC video transmission strategy is better than that of greedy strategy. The average playback quality with cluster-based cooperation is better than that with one-hop V2V cooperation, because the target and its neighboring vehicles can form a cluster to receive data for taking full advantage of the wireless resource. What's more, the average playback quality will decrease with the increase of inter-RSU distance, as the amount of data forwarded by carriers is smaller than that by RSU. Fig. \ref{Ave_qua} shows the performance of average quality variation. As we mentioned above, our BC video transmission strategy can guarantee AQV is equal to $L$ at most, yet the AQV with greedy strategy will increase with the increase of inter-RSU distance. 

Fig. \ref{throughput} shows the impact of inter-RSU distance on throughput. The throughput decreases with the increase of inter-RSU distance and the throughput with cooperation scheme is larger than that with relay-aided scheme, which have the same reason as the curves in Fig. \ref{Ave_qua}. In addition, the simulation results are very close to analytical results, showing that our analytical method is effective and the analytical results are accurate.
\section{Conclusions}

In this paper, We proposed a cooperative communication scheme to facilitate video data transmission. With the knowledge of vehicle status information, we designed a BC video transmission strategy to decrease the interruption ratio, video quality variation and improve the playback quality. Simulation results demonstrated that our approach can achieve better performance compared with relay-aided scheme and greedy transmission strategy. In addition, by extending the V2V communication time, the performance of target-cluster-based cooperation is better than that of one-hop V2V cooperation. We also analyze the throughput with one-hop and target-cluster-based cooperation schemes respectively and obtain their closed-form expressions, which is verified through simulations. The value of throughput is useful to the central server for video encoding design. In the future, we will design a video transmission strategy considering packet loss.

\bibliographystyle{IEEEtran}
\bibliography{IEEEabrv,reference}

\end{document}